\documentclass{article15}
\usepackage[]{algorithm2e}
\usepackage{filecontents}
\usepackage[nottoc]{tocbibind}
\usepackage{cite}
\usepackage{graphicx}
\usepackage{wrapfig}
\usepackage{setspace}
\usepackage{lipsum}

\title{Geometric Embedding of Path and Cycle Graphs in Pseudo-convex Polygons}
\author{Hamid Hoorfar\thanks{Department of Computer Engineering and Information Technology, Amirkabir University of Technology (Tehran Polytechnic), {\tt \{hoorfar,ar.bagheri\}@aut.ac.ir}} \and Alireza Bagheri\footnotemark[1]}

\begin{document}
\maketitle

\begin{abstract}
	
	Given a graph $ G $ with $ n $ vertices and a set $ S $ of $ n $ points in the plane, a point-set embedding of $ G $ on $ S $ is a planar drawing such that each vertex of $ G $ is mapped to a distinct point of $ S $. A straight-line point-set embedding is a point-set embedding with no edge bends or curves. The point-set embeddability problem is NP-complete, even when $ G $ is $ 2 $-connected and $ 2 $-outerplanar. It has been solved polynomially only for a few classes of planar graphs. Suppose that $ S $ is the set of vertices of a simple polygon. A straight-line polygon embedding of a graph is a straight-line point-set embedding of the graph onto the vertices of the polygon with no crossing between edges of graph and the edges of polygon. In this paper, we present $ O(n) $-time algorithms for polygon embedding of path and cycle graphs in simple convex polygon and same time algorithms for  polygon embedding of path and cycle graphs in a large type of simple polygons where $n$ is the number of vertices of the polygon.\\
	
\textbf{keywords:}
	graph drawing, point-set embedding, Halin graph, fan graph.
\end{abstract}

\section{Introduction}
Geometric embedding of graphs has wide applications in circuit schematics, algorithm animation, and software engineering\cite{di1994algorithms}. This problem has been studied extensively for planar graphs in various fields. Planar graph drawing is the problem of drawing a graph in the plane such that its edges intersect only at common vertices. It has been proved that any planar graph can be drawn without crossing, even when its edges are drawn as a straight-line segment joining its vertices.
 Schnyder\cite{schnyder1990embedding} showed that each plane graph with $n\ge 3$ vertices has a straight line embedding on the $\left(n-1\right)\times\left(n-1\right)$ grid. He proved that the embedding is computable in time $ O(n) $. In the point-set embeddability problem given a planar graph $ G $ and a set $ P $ of points in the plane, deciding whether there is a planar straight-line embedding of $ G $ such that the  vertices are embedded onto the points $ P $ is NP-complete, even when $ G $ is $ 2 $-connected and $ 2 $-outerplanar\cite{cabello2006planar}. NP-hardness  for $ 3 $-connected graphs was shown by Durocher and Mondal\cite{durocher2012hardness}. For outer-planar graphs was presented an $ O(n log^{3} n) $ time and $ O(n) $ space algorithm to compute a straight-line embedding of $ G $ in $ P $\cite{bose2002embedding,castaneda1996straight,pach1991embedding}.There is an $\Omega \left(n\mathrm{log}n\right)$ lower bound for the point-set embeddability problem\cite{bose1997optimal}. Several papers studied point-set embeddability for special graph classes. Nishat et al.\cite{nishat2012point}, Moosa and Rahman\cite{moosa2011improved} and Hosseinzadegan and Bagheri\cite{Hoseinzadegan} presented algorithms to test point-set embeddability of triangulated planar graphs of treewidth $ 3 $ and wheel graphs. Optimal embedding of a planar graph $ G $ means an embedding of $ G $ such that the total edge length of embedding is the minimum. Optimal embedding application is in the design of distributed computing networks, a special case of the problem is the Euclidean Travelling Salesman Problem (TSP), where graph $ G $ is a cycle of $ n $ nodes\cite{Hoseinzadegan}. 
 There is an algorithm for embedding rooted trees and degree-constrained trees in arbitrary point set, computable in $ O(n \log n) $-time\cite{ikebe1994rooted}. A wheel graph $ W_{n} $ is a graph with $ n $ vertices, formed by connecting a single vertex to all vertices of an $ (n-1) $-cycle.
 The single central vertex will be referred to as \textit{hub} and others vertices as \textit{rim} vertices\cite{yang2009beyond}. Chambers et al\cite{chambers2010drawing} describe an algorithm for drawing a planar graph with a prescribed outer face
 shape. The input consists of an embedded planar graph $ G $, a partition of the outer face of the embedding into a set $ S $ of $ k $ chord-free paths, and a $ k $-sided polygon $ P $; the output of their algorithm is a drawing of $ G $ within $ P $ with each path in $ S $ drawn along an edge of $ P $. In our problem, we are not allowed to draw any edge of $ G $ along an edge of $ P $, even with intersection and we should not use any point, except vertices of $ P $, for graph embedding. In similar problem by\cite{mchedlidze2013drawing}. given a planar graph $ G $ with a fixed planar embedding and a simple cycle $ C $ in $ G $ whose vertices are in convex position, they studied the question whether this drawing can be extended to a planar straight-line drawing of $ G $. they characterize when this is possible in terms of simple necessary and sufficient conditions. For this purpose, they described a linear-time testing algorithm. Hoseinzadegan and Bagheri\cite{Hoseinzadegan} consider optimal embedding of wheel graphs and a sub-class of 3-trees, that are not outer-planar. they presented optimal $ O(n \log n) $-time algorithms for embedding and $ O(n^{2}) $-time algorithms for optimal embedding of wheel graphs. Wheel graphs are a sub-class of Halin graphs. A Halin graph $ H $,also known as a roofless polyhedron, is obtained by a planar drawing of a tree having four or more vertices, having no nodes of degree $ 2 $ in the plane, and then connecting all leaves of the tree with a cycle $ C $ which passes around the tree's boundary in such a way that the resulting graph is planar. Halin graphs are edge $ 3 $-connected and Hamiltonian\cite{cornuejols1983halin}.
 \begin{figure}
 \centering
 \includegraphics[width=0.7\linewidth]{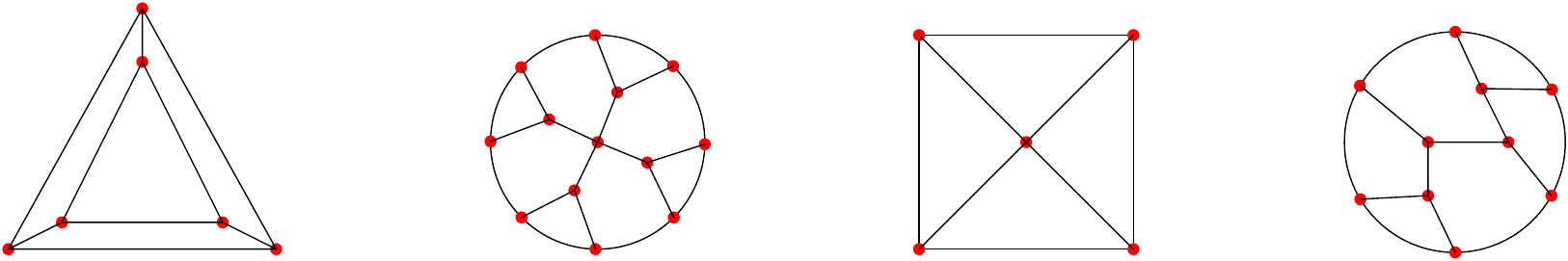}
 \caption{{\small Some examples of Halin graphs}}
 \label{fig:figure02}
 \end{figure}
 A gear graph is obtained from the wheel $ W_{n} $ by adding a vertex between every pair of adjacent vertices of the $ n $-cycle\cite{gallian2014dynamic}. In 2010, Bagheri and Razzazi\cite{bagheri2010planar} showed that the problem of deciding whether there is a planar straight-line point-set embedding of an $ n $-node tree $ T $ on a set $ P $ of $ n $ points in the plane that includes a partial embedding $ E $ of $ T $ on $ P $ is NP-complete. In another paper\cite{sepehri2010point}, Sepehri and Bagheri presented an algorithm for embedding a tree with $ n $ vertices on a set of $ n $ points inside a simple polygon with $ m $ vertices so that number of bends is minimum. They showed that time complexity of the presented algorithm is $ O(n^{2}m+n^{4}) $. A geometric graph H is a graph $ G(H) $ together with an injective mapping of its vertices into the plane. An edge of the graph is drawn as a straight-line segment joining its vertices. We use $  V(H) $ for the set of points where the vertices of $  G(H) $ are mapped to, and we do not make a distinction between the edges of $ G(H) $ and $  H $. A planar geometric graph is a geometric graph such that its edges intersect only at common vertices. In this case, we say that H is a geometric planar embedding of $ G(H) $. In the next section, we present some definitions and preliminaries that that will be used throughout the paper.
 
\section{Preliminaries}
We use basic notions of graph drawing that is used in~\cite{di2009point}. Let $ G = ( V , E ) $ be a simple graph with $ n $ vertices in the node set $ V $ and $ E $ be the edge set of $ G $. A $ planar graph $ is a graph with at least
a drawing without edge crossing except at the nodes, where the edges are incident. A \textit{plane graph} is a planar graph with a fixed embedding on the plane. let $ S $ be a set of $ n $ points in the plane. A point-set embedding of $ G $ onto $ S $ , denoted as $\varGamma ( G , S ) $, is a planar drawing of $ G $ such that each vertex is mapped to a distinct point of $ S $ . $ \varGamma ( G , S ) $ is called a geometric (straight-line) point-set embedding if each edge is drawn as a straight-line segment. Let $ D ( S ) $ be a straight-line drawing whose vertices are points of a subset of $ S $. An optimal point-set embedding of a planar graph $ G $ is a point-set embedding of $ G $, where the total edge length
(or the area) of the embedding is minimized.  A path graph $ P_{n} $ is a tree with $ n-2 $ nodes of degree $ 2 $ and two nodes of degree one. A cycle graph $ C_{n} $, known as an $ n $-cycle is a graph with $ n $ nodes containing a single cycle through all nodes. In cycle graph $ C_{n} $, every nodes are of degree $ 2 $. A complete graph $ K_{n} $ is a graph with $ n $ nodes and $ \frac{n\left(n-1\right)}{2} $ edges, in which each pair of graph vertices is connected by an edge. A complete graph that is embedded in an other graph is called clique. A wheel graph $ W_{n} $, is a planar graph with at least $ 3 $ nodes, such that a certain node is connected to all other nodes of a $ (n-1) $-cycle. Wheel graph is extended to Halin graph. A Halin graph (so called roofless polyhedron) is a planar graph formed of a tree $ T $, without any $ 2 $-degree nodes, and a cycle $ C $ connecting leaves of $ T $ in the cyclic order determined by a plane embedding of $ T $. Therefore, A Halin Graph $ H_{n,m} $ has $ n+m $ nodes, $ n $ of them belongs to tree $ T $ and $ m $ of them belongs to cycle $ C $. If $ T $ be a star, Halin graph becomes a wheel graph. A fan graph $ F_{n,m} $ is defined as a path graph $ P_{m} $ and empty graph $ E_{n} $ with $ n $ nodes such that every node belongs $ E_{n} $ is connected to every node belong $ P_{m} $ with additional edges. The convex hull of a point set $ S $ is a minimum area convex polygon $ CH(S) $ that is contain all points of $ P $. Convex hull $ CH(S) $ is a partition of set $ S $ into two subsets are named \textit{inner points} and \textit{outer points}. The points that is placed on boundary of $ CH(S) $ belong to outer points set and other points that is placed inside $ CH(S) $(without boundary) belong to inner points set, denoted as $ S_{out} $, $ S_{in} $, respectively. It is obvious to get the following result, $\left|S\right|\ge\left|{S}_{out}\right|\ge 3$. Convex hull of a point set with $ n $ points can be compute in $O\left(n\mathrm{log}n\right)$ time~\cite{de2000computational}.
We introduce new version of graph drawing, named as \textit{polygon embedding} . Let $ G=(V,E) $ be a simple graph with $ n $ vertices in the node set $ V $ and $ E $ be the edge set of $ G $. Furthermore, Let $ P=(\varSigma,\varPi) $ be a simple polygon with $m(\ge n)$ vertices in the vertex set $ \varSigma $ and $ \varPi $ be the edge set of $ P $. A \textit{polygon embedding} of $ G $ on $ P $ is drawing of $ G $ such that each vertex of $ G $ mapped to a distinct vertex of $ P $ without edge crossing between edges in $ E$ and edges in $\Pi$ except at the nodes, where the edges are incident and $ E\cap \Pi =\phi $. For each edge $ e $ of $ G $, must be $ e\subseteq P $. See figure~\ref{fig:figure04}. A polygon embedding of $ G $ into $ P $ is named \textit{planar} such that drawing be without edge crossing between edges of graph $ G $ and called a \textit{geometric (straight-line) polygon embedding} if each edge is drawn as a straight-line segment.
\begin{figure}
 \centering
 \includegraphics[width=0.5\linewidth]{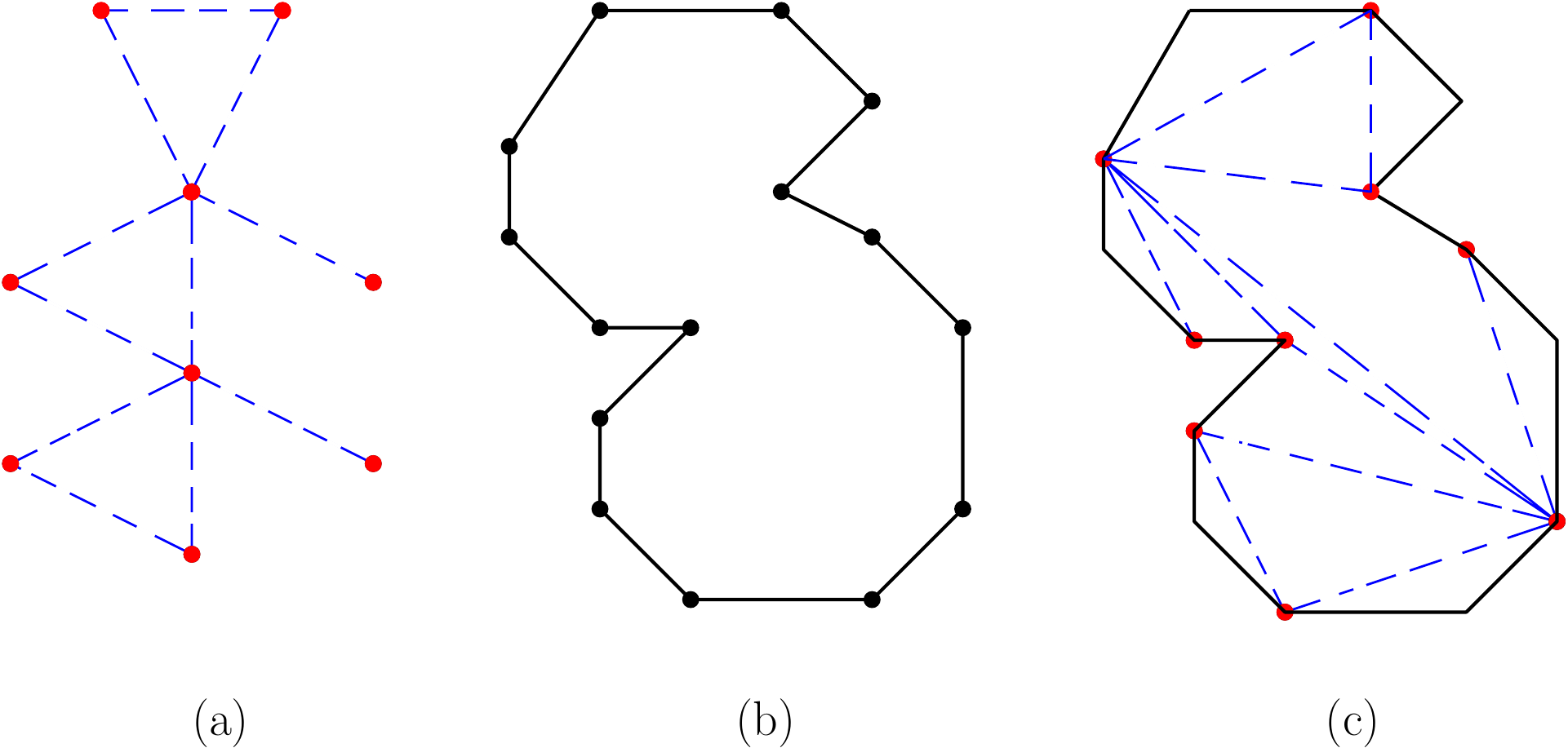}
 \caption{\small{ A \textit{polygon embedding} of $ G $ into $ P $}}
 \label{fig:figure04}
 \end{figure}
Polygon embedding is a special case of point set embedding with some additional constraints which make problem more complex than before. For example, connecting between some pair of points in $ S $ (vertices of polygon) is not allowed. In many cases, polygon embedding is possible if number of vertices of $ P $ be greater than number of vertices of $ G $. Therefore, always mapping must be one-to-one but sometimes not onto.
There is an $ \varOmega(n\log n) $ lower bound for the point-set embeddability problem. For path and cycle graphs was presented an $ \varTheta(n\log n) $ time and $ O(n) $ space algorithm to compute a straight-line embedding in point set $ S $. We will review optimum-time algorithms to embedding these two kinds of graphs in a given point set in the next section. After that, we will present algorithms for straight-line polygon embedding of path and cycle graphs in polygon $ P $ and compute upper bound on maximum edges of graph that can be embeddable in $ P $. A convex polygon is a polygon with all its interior angles less than $ \pi $. A concave polygon is a polygon that is not convex. A simple polygon is concave if at least one of its internal angles is greater than $ \pi $. An orthogonal polygon is one whose edges are all aligned with a pair of orthogonal coordinate axes, which we take to be horizontal and vertical without loss of generality. Thus, the edges alternate between horizontal and vertical, and always meet orthogonally, with internal angles of either $ \frac{\pi }{2} $ or $ \frac{3\pi }{2} $~\cite{o1987art}.  An orthogonal polygon is defined \textit{orthoconvex} if the intersection of the polygon with a horizontal or a vertical line is a single line segment~\cite{nandy2010recognition}. Orthoconvex with $ n\geqslant 6 $ vertices is a concave polygon. Every orthogonal polygon with $ n $ vertices has $ \frac{n-4}{2} $ angles of $ \frac{3\pi }{2} $ and $ \frac{n+4}{2} $ angles of $ \frac{\pi }{2} $, exactly. 
Let $ p $ and $ q $ are two points in polygon $ P $. $ p $ and $ q $ are said \textit{visible} to each other, if the line segment that joins them does not intersect any edge of $ P $. So, $ p $ and $ q $ are said to be \textit{invisible} to each other, if they are not visible. A \textit{visibility graph} is a graph for polygon $ P $, as denoted $ G_{v}^{P}(V,E) $, that each node in $ V $ represents a point location in $ P $, and there is an edge between two nodes $ v_{i} $ and $ v_{j} $ if they are visible from each other. A visibility graph $ G_{v}^{P}(V,E) $ is named \textit{polygon visibility graph}, if $ V$ be vertex set of $ P $. Note that in the whole of paper, size of graph means number of its edges.

\section{Preliminary Algorithms and Results}
\label{s:s03}
In the following, we review a exact $ O(n \log n) $-algorithm for embedding of path graph $ P_{n} $ onto point set $ S $ with $ n $ points which is always possible. First, order point set $ S $ according to $ X $-axes in list $ L=\left[{p}_{0},{p}_{1},{p}_{2,\dots },{p}_{n}\right] $, ascending. Then, connect points according to this ordering successive. There is a similar $ O(n \log n) $-algorithm for embedding of cycle graph $ C_{n} $ onto point set $ S $, as following:  
\begin{enumerate}
  \item Order point set $ S $ according to $ X $-axes in list $ L=\left[{p}_{0},{p}_{1},{p}_{2,\dots },{p}_{n-1}\right] $, ascending.
  \item Make list $ L_{down}=\left[{d}_{0},{d}_{1},{d}_{2,\dots },{d}_{\frac{n}{2}-1}\right] $ such that $ {d}_{i}=\underset{y}{\mathrm{min}}\left\{{p}_{2i},{p}_{2i-1}|i< \frac{n}{2},i\in N\right\} $.  
  \item Make list $ L_{up}=\left[{u}_{0},{u}_{1},{u}_{2,\dots },{u}_{\frac{n}{2}-1}\right] $ such that $ {u}_{i}=\underset{y}{\mathrm{max}}\left\{{p}_{2i},{p}_{2i-1}|i< \frac{n}{2},i\in N\right\} $.
  \item Connect points in each list according to their ordering successive. 
  \item Connect $ d_{0} $ to $ u_{0} $ and $ d_{\frac{n}{2}-1} $ to $ u_{\frac{n}{2}-1} $.
\end{enumerate}
 \begin{figure}
  \centering
  \includegraphics[width=0.6\linewidth]{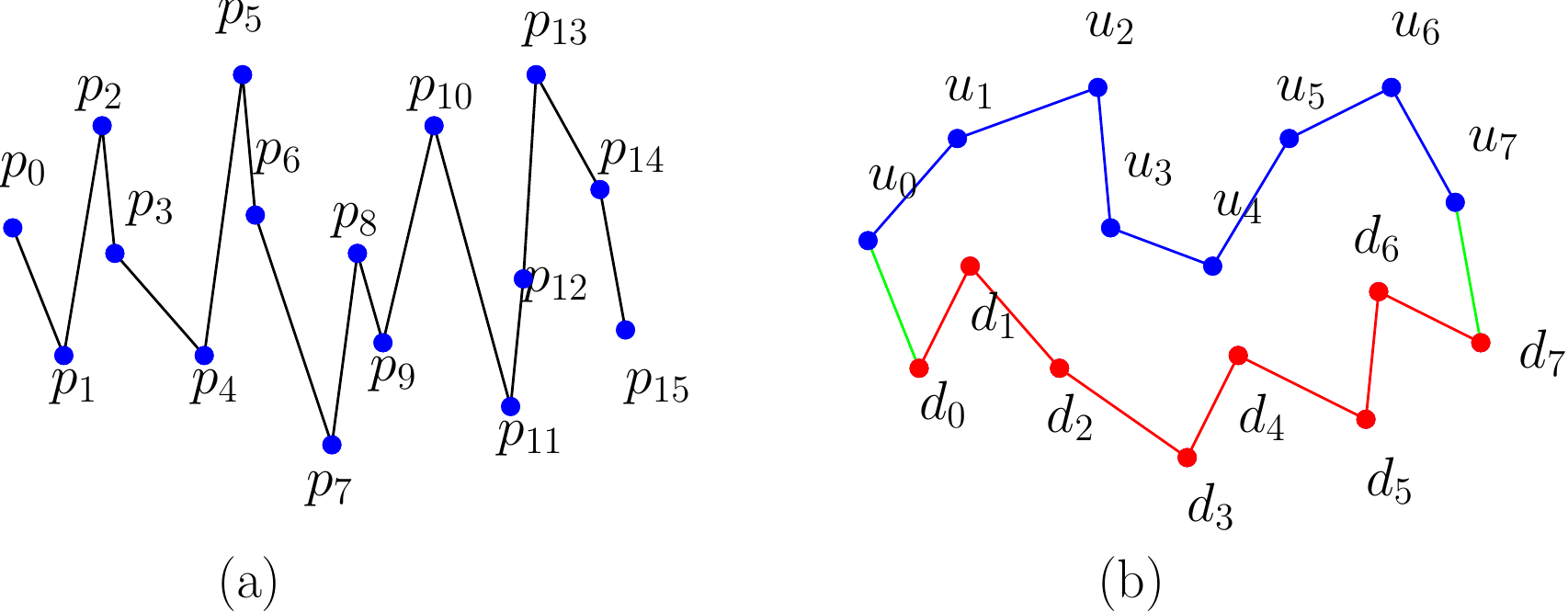}
  \caption{{\small (a)Embedding of path graph and (b)cycle graph in the point set.}}
  \label{fig:figure05}
  \end{figure}
See figure~\ref{fig:figure05}. These graph embeddings are not unique and optimum. Time complexity of both algorithms is equal to time complexity of sorting algorithm that puts elements of list in a certain order. Therefore, time complexity is $ \varTheta(n \log n) $ and it is tight. Maximum path(cycle)that can be draw in point set of size $ n $ is $ n $, but where we want to embed a path (cycle) graph in polygon with $ n $ vertices, maximum size is not $ n $. We show that maximum size of path graph that can be embedded planar in polygon with $ n $ vertices is not greater than $ n-3 $ and for cycle is not greater than $ \lfloor\frac{n}{2}\rfloor $, where polygon is convex. Therefore, we have the following obviously results. Given a polygon $ P $ with $ n $ vertices, maximum size of path graph that can be embedded planar in $ P $ is lesser than or equal to $ n-3 $. Edges of path graph that can be embedded planar in $ P $ must be chords of $ P $ without crossing except at their end points. Maximum number of chords in a polygon with $ n $ vertices without crossing is $ n-3 $. Therefore, maximum size of path graph that can be embedded planar in $ P $ can not be greater than $ n-3 $. Also, maximum size of every graphs that can be embedded planar in polygon $ P $ with $ n $ vertices is $ n-3 $. By the way, given a polygon $ P $ with $ n $ vertices, maximum size of cycle graph that can be embedded planar in $ P $ is  lesser than or equal to $ \lfloor\frac{n}{2}\rfloor $. Let cycle graph $ C $ is embedded in $ P $ and $ \{c_{i}|0\leq i \leq m\} $ is set of its nodes, ordered counter clockwise. Each node $ c_{i} $ of $ C $ is degree two and must be connected to two different other nodes {\small (as named $ c_{i-1} $ and $ c_{i+1} $)} which are vertices of $ P $. Let nodes $ c_{i-1}$, $c_{i} $ and $ c_{i+1} $ of cycle are mapped to vertices $ p_{j}$, $ p_{k} $ and $ p_{l} $ in $ P $ {\small (ordered counter clockwise)}. Therefore, ($ p_{j}$,$ p_{k} $) and ($ p_{k} $,$ p_{l} $) can not  be edges of $ P $ because edges of cycle are not allowed to be edges of $P$. Cycle $ C $ is planar, so, any other nodes of $ C $ can not mapped to vertices of $ P $ between $ p_{j}$, $ p_{k} $ {\small (as ordered, $ {p}_{j+1 },\dots,{p}_{k-1} $)} and between $ p_{k} $, $ p_{l} $ {\small (as ordered, $ {p}_{k+1 },\dots,{p}_{l-1} $}), see figure~\ref{fig:figure06}. Suppose node $ {c}_{f} $ is mapped to $ {p}_{g} $ and {\small $ j<g<k $}, so, edge ($ c_{f-1}$,$ c_{f} $) and ($ c_{f} $,$ c_{f+1} $) of $ C $ must cross edge ($ c_{i-1} $,$ c_{i} $) and It is contradiction. In the other words, if nodes $ c_{i-1}$, $c_{i} $ and $ c_{i+1} $ of cycle are mapped to vertices $ p_{j}$, $ p_{k} $ and $ p_{l} $ in $ P $, consecutively, then vertices $ {p}_{j+1 },\dots,{p}_{k-1} $ and $ {p}_{k+1 },\dots,{p}_{l-1} $ are out of reach to be in cycle $ C $. ($ p_{j}$,$ p_{k} $) and ($ p_{k} $,$ p_{l} $) are not adjacent in $ P $. Therefore, at least one vertex between $ p_{j}$ and $ p_{k} $ and one vertex between $ p_{k} $ and $ p_{l} $ are out of reach to be in planar cycle $ C $. It happens for every $ c_{i}$, {\small $0\leq i \leq m $}. Hence, for every two successive nodes in $ S $ at least one vertex of $ P $ must be blocked and never can be mapped to any node of $ C $. Maximum size of $ C $ is lesser than or equal to $ \lfloor\frac{n}{2}\rfloor $.
 \begin{figure}
  \centering
  \includegraphics[width=0.3\linewidth]{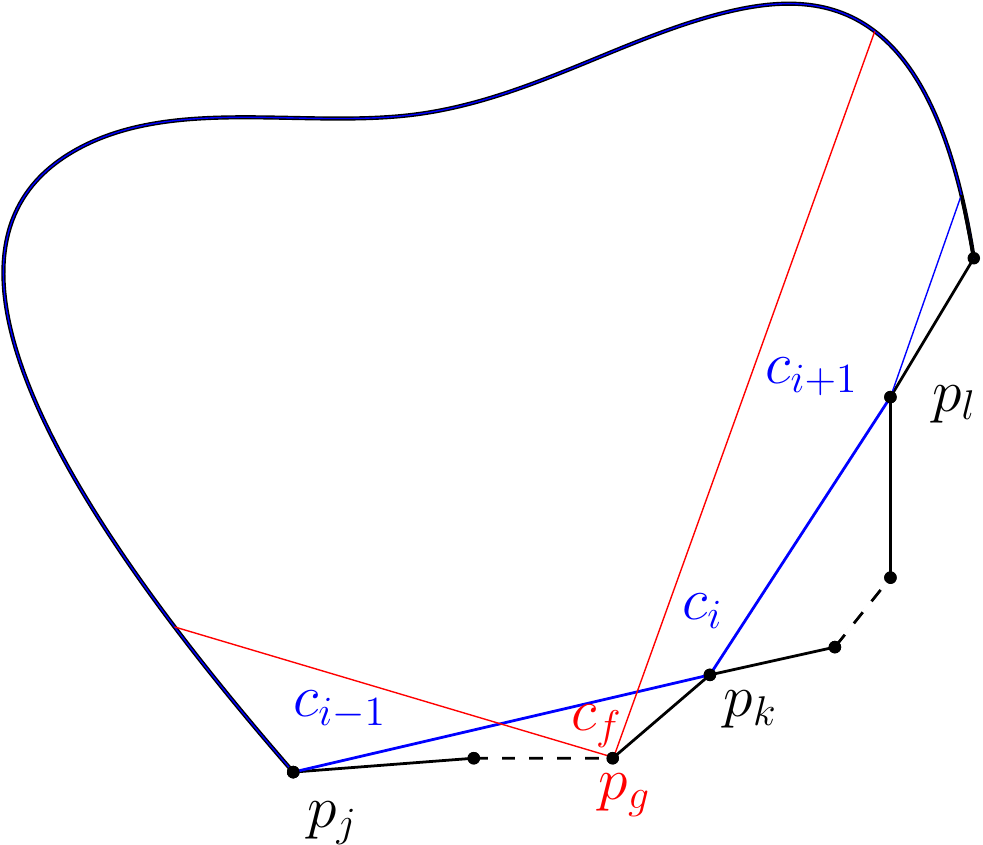}
  \caption{{\small Embedding of cycle graph in the polygon.}}
  \label{fig:figure06}
  \end{figure}
Also, maximum clique that can be embedding in polygon $ P $ with $ n $ vertices such that nodes of clique are mapped to vertices of $ P $ is $ K_{\lfloor \frac{n}{2}\rfloor}$, where $ P $ is convex. There is a linear-time reduction from sorting problem to point-set embeddability problem that be used to show the second problem is at least as difficult as the first. For time complexity, $ \varOmega(n\log n) $ is lower bound for the point-set embeddability problem. There are several algorithms for embedding trees in arbitrary point set, computable in $ O(n \log n)$-time\cite{ikebe1994rooted}. All of them using a sorting algorithm on point set and ordering it corresponding to a axes line or angular. In the polygon embeddability problem, given point set is placed on vertices set. The vertices of a polygon have a kind of ordering by itself. We study the following lemmas. 
\begin{lemma}
	\label{le:le03}
	There is a linear-time algorithm for straight-line planar embedding of a path graph $ P_{m} $ into convex polygon $ Q_{m+3} $ such that each node of $ P_{m} $ mapped to a vertex of $ C $.
\end{lemma}
\begin{proof}
\label{pr:pr03}
Let $ \{q_{0},q_{1},\dots,q_{m+2}\} $ be the set of vertices of $ Q $ ordered counter clockwise. Traverse vertices of $ Q $ in following order:
\begin{enumerate}
\item Traverse from $ {q}_{m+2} $ to $ {q}_{m} $.
\item Traverse from $ {q}_{m-i} $ to $ {q}_{i} $ {\small (for $ 0\leq i<\frac{m}{2} $)} if current location index is greater than $ \frac{m}{2} $.
\item Traverse from $ {q}_{i} $ to $ {q}_{m-i-1} $ {\small (for $ 0\leq i< \frac{m}{2}-1 $)} if current location index is lesser than $ \frac{m}{2} $.
  \end{enumerate}
 Because of convexity, vertices of convex polygon are visible from each others. Therefore, according to algorithm, path $ P $ has exactly $ m $ edges which are all planar chords of $ Q $. Time complexity is $ \theta \left(m\right) $, see figure~\ref{fig:figure07}(a).
\begin{figure}
 \centering
 \includegraphics[width=0.5\linewidth]{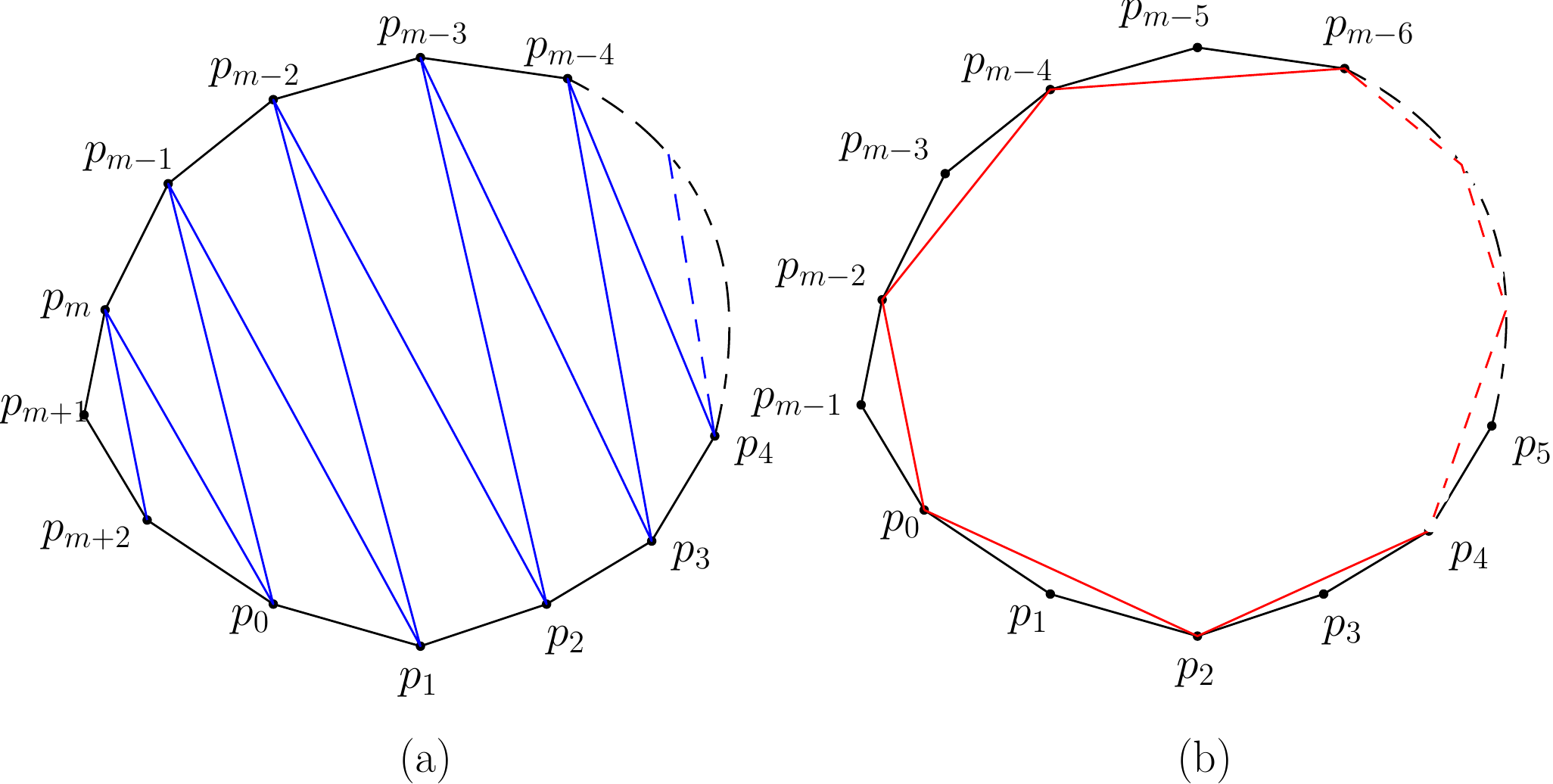}
 \caption{{\small Embedding of maximum (a)path and (b)cycle in convex polygon.}}
 \label{fig:figure07}
 \end{figure}
\end{proof}
There is an $ \varOmega(n) $ lower bound on time complexity for the polygon embeddability problem(in comparison with lower bound for point set embeddability that was $ \varOmega(n \log n) $ is better).
\begin{lemma}
	\label{le:le04}
	There is a linear-time algorithm for straight-line planar embedding of a cycle graph $ C_{\lfloor \frac{m}{2}\rfloor} $ into convex polygon $ P_{m} $ such that each node of $ C_{m} $ mapped to a vertex of $ P $.
\end{lemma}
\begin{proof}
There are similarities between this proof and previous one. Let $ \{p_{0},p_{1},\dots,p_{m-1}\} $ be vertices set of $ P_{m} $ ordered counter clockwise. Tour vertices of $ P_{m} $ in following order:
\begin{enumerate}
\item Traverse from $ {p}_{2i} $ to $ {p}_{2i+2} $, {\small for $ 0\le i <\lfloor\frac{m}{2}\rfloor-1 $}.
\item Turn back from $ {p}_{2\lfloor\frac{m}{2}\rfloor-2} $ to $ {p}_{0}$.
  \end{enumerate}
 Because of convexity, vertices of convex polygon are visible from each others. Therefore, according to algorithm, $ C $ has exactly $ \lfloor\frac{m}{2}\rfloor $ edges which are all planar chords of $ P $. Time complexity is $ \theta \left(m\right) $, see figure~\ref{fig:figure07}(b).
\end{proof}
Our algorithms for embedding path and cycle use segments from the polygon chords to avoid intersections between embedded edges. 

\section{Embedding a cycle graph with maximum size in pseudo-convex polygon}
In this section, we first definite \textit{pseudo-convex} polygon and provide a linear-time algorithm for straight-line planar embedding of maximum cycle $ C_{max} $ in a pseudo-convex polygon $ P $ with $ n $ vertices. If polygon $ P $ is not convex, then there are some pairs of vertices that is not visible from each others. Therefore, let $ (v_{i},v_{j}) $ be an invisible pair of vertices, no edge of graph can be mapped to it. In the other words, finding maximum cycle in $ P $ is reduced to finding maximum cycle of polygon visibility graph of $ P $. Visibility graph of a polygon computes in time of $ O(n^{2}) $. In addition, finding maximum cycle in a graph is known as NP-hard. 
\begin{defini}
A reflex vertex in polygon $ P $ that is adjacent with another reflex is named \textit{u-turn vertex} and an edge in $ P $ that its both endpoints are u-turn vertices is named \textit{u-turn edge}, see figure~\ref{fig:figure08}(a).
\end{defini}
\begin{defini}
Let $ P $ be a polygon and $ V=\left\{{p}_{0},{p}_{1},{p}_{2},\dots ,{p}_{n-1}\right\} $ be vertex set of $ P $ ordered counter clockwise, if there is no successive reflex vertices in $ V $, $ P $ is named \textit{pseudo-convex polygon}. In the other word, $ P $ is pseudo-convex polygon if it has no u-turn vertex (edge), see figure~\ref{fig:figure08}(b). 
\end{defini}
\begin{defini}
Let $ P $ be a polygon, $ V=\left\{{p}_{0},{p}_{1},{p}_{2},\dots ,{p}_{n-1}\right\} $ be vertex set of $ P $ ordered counter clockwise and $ S_{i} $ be the set of vertices that are visible from $ {p}_{i} $. $ {p}_{i} $ is named \textit{isolated vertex}, if it has one of the following conditions:
\begin{itemize}
\item Cardinality of $ S_{i} $ be lesser than five (note that $ {p}_{i}\in S_{i} $). \item Cardinality of $ S_{i} $ be equal to five and both members of $ S_{i}-\{p_{i-1},p_{i},p_{i+1}\} $ be adjacent.
\end{itemize} 
\end{defini}
 \begin{figure}
  \centering
  \includegraphics[width=0.7\linewidth]{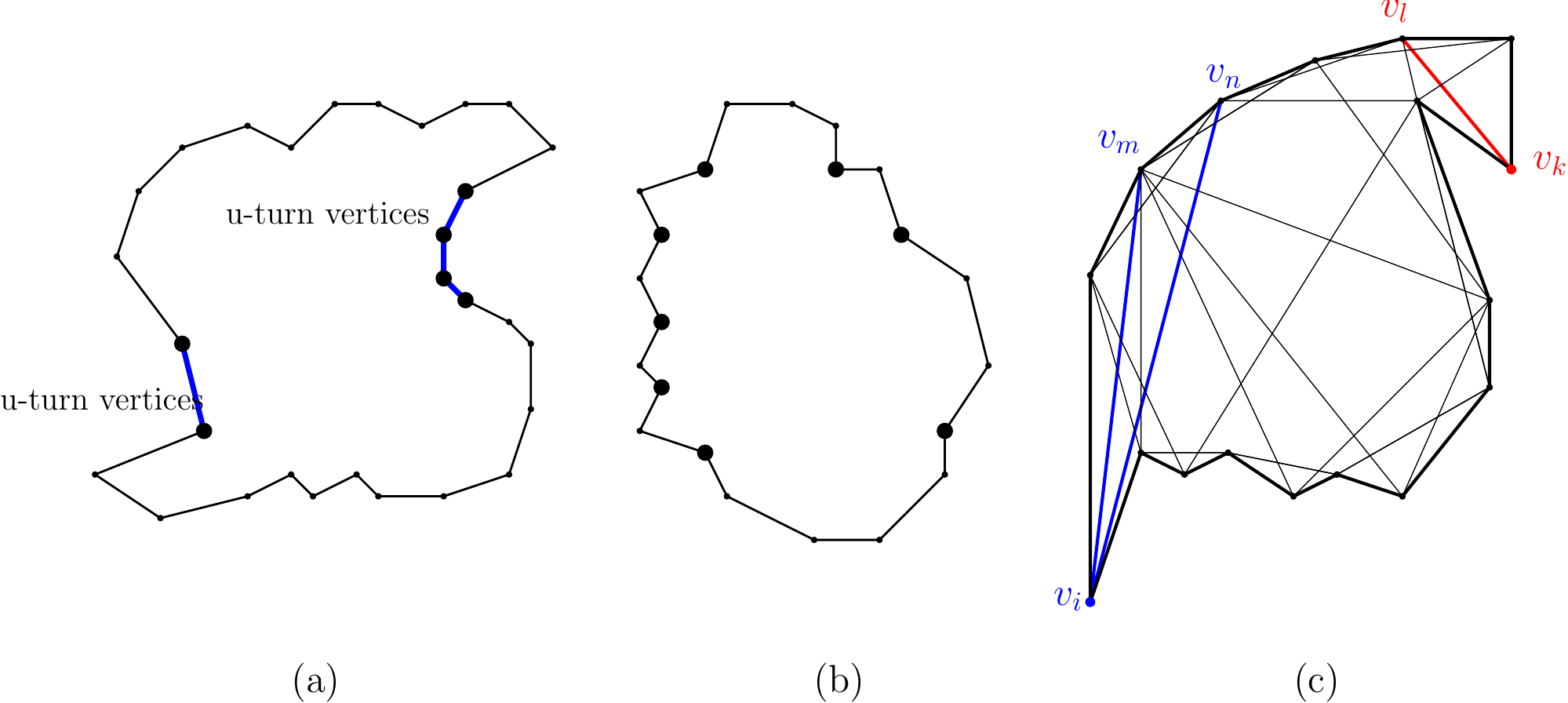}
  \caption{{\small (a)A polygon with u-turn vertices and edges: bold disks are u-turned vertices and fat segment are u-turned edges. (b)A pseudo-convex polygon: bold disks are reflex vertices. There are no successive reflex vertices. (c)Vertices $v_{i}$ and $v_{k}$ can not be a part of an embedded cycle graph but other vertices can be  at least in one cycle those are drawn by the cycles.}}
  \label{fig:figure08}
  \end{figure}
See figure~\ref{fig:figure08}(c), $ v_{i} $ and $ v_{k} $ are isolated vertices. Convex and orthoconvex polygon are two examples for pseudo-convex polygon. If a polygon has a u-turn vertex, actually, it has two successive u-turn vertices or more. In the following, we provide a linear-time algorithm for straight-line embedding of maximum cycle graph $ C $ in a pseudo-convex polygon $ P $. A vertex $ v_{i} $ of $ P $ can be in the cycle, if it is visible from at least two different vertices which are not adjacent to $ v_{i} $. Because the degree of every node $ c_{j} $ in the cycle $ C $ is two, that means, it must be connect to two others nodes $ c_{j-1} $ and $ c_{j+1} $. Therefore, if $ c_{j} $ is mapped to $ v_{i} $, $ P $ must have two chords from $ v_{i} $ to other vertices which the edges must be mapped to. Otherwise straight-line embedding is not possible by this mapping. Hence, a mapped vertex $ v_{i} $ in every embedded cycle must be visible from at least two other non-adjacent vertices $ v_{s} $ and $ v_{t} $ of $ P $. In addition, if $ v_{s} $ and $ v_{t} $ are adjacent, being part of a cycle is not still possible. There is no isolated vertex of $ P $ in cycle graph $ C $. See figure~\ref{fig:figure08}(c).\\
In the following, we want to study straight-line planar embedding of a cycle graph with maximum size in pseudo-convex polygon. Let $ C_{m} $ be a cycle graph with maximum number of edges that is embedded in pseudo-convex polygon $ P_{n} $ with $ n $ vertices. Let $ \{p_{0},p_{1},\dots,p_{n-1}\} $ be vertex set of $ P_{n} $ ordered counter clockwise and $ \{c_{0},c_{1},\dots,c_{n-1}\} $ be node set of $ C_{m} $ ordered counter clockwise. If node $ c_{i} $ is mapped to vertex $ p_{j} $, then $ p_{j-1} $ and $ p_{j+1} $ can not be nodes in $ C $. If node $ c_{i} $ is mapped to vertex $ p_{j} $ and vertex $ p_{j+1} $ is a reflex vertices, then it is better that $ c_{i} $ be mapped to vertex $ p_{j+1} $ instead of $ p_{j}$, and size of $ C $ will not change and still remain maximum. We must probe this explained claim. Consequently, we prove the following lemmas.
\begin{lemma}
Let $ p_{j-1} $, $ p_{j} $ and $ p_{j+1} $ be three successive non-isolated vertices in pseudo-convex polygon $ P $. There is at least one node in maximum embedded cycle that is mapped to one of these three successive vertices.
\end{lemma}
\begin{proof}
Suppose that there is a maximum embedded cycle $ C $ such that any node of $ C $ is not mapped to these three successive vertices $ p_{j-1} $, $ p_{j} $ and $ p_{j+1} $. Let $ p_{s} $ be the first previous vertex (first before $ p_{j-1} $) such that a node of $ C $ is mapped to it and $ p_{t} $ be the first next vertex (first after $ p_{j+1} $) such that a node of $ C $ is mapped to it. There must exist a straight-line planar polygon embedded path, as denoted $ \varPi $, between $ p_{s} $ and $ p_{t} $ that is contained at least one of $ p_{j-1} $, $ p_{j} $ or $ p_{j+1} $ without edge crossing with cycle edges except at the nodes $ p_{s} $ and $ p_{t} $. The size of $ \varPi $ is greater than one, if we remove edge $( p_{s} , p_{t}) $ from $ C $ and replace path $ \varPi $ instead of it, we should have a new cycle that its size is greater than size of $ C $. It means $ C $ was not maximum cycle and it is contradiction.
\end{proof}
\begin{cor}
At least one and at most two vertices of every three successive non-isolated vertices in pseudo-convex polygon $ P $ are on a cycle graph with maximum number of edges that is embedded in $ P $. 
\end{cor}
\begin{lemma}
\label{le:le05}
There is a straight-line planar embedded cycle graph with maximum size such that all reflex vertices of $ P_{n} $ are on it, where $ P_{n} $ is a pseudo-convex polygon.
\end{lemma}
\begin{proof}
Let $ C_{m} $ be a cycle graph with maximum number of edges that is embedded in pseudo-convex polygon $ P_{n} $ with $ n $ vertices. Let $ \{p_{0},p_{1},\dots,p_{n-1}\} $ be vertex set of $ P_{n} $ ordered counter clockwise and $ \{c_{0},c_{1},\dots,c_{n-1}\} $ be node set of $ C_{m} $ ordered counter clockwise. Cycle graph $ C_{m} $ has $ m $ nodes that is mapped to $ m $ vertices of $ P_{n} $. Suppose that $ p_{i} $ is a reflex vertex and any node of $ C_{m} $ is not mapped to it. So, if $ p_{i} $ is a reflex vertex, then $ p_{i-1} $ and $ p_{i+1} $ must be convex vertices, we find it from being pseudo-convexity.  In the first case, let there is a node, as denoted $ c_{j} $, that is mapped to $ p_{i-1} $. So, $ c_{j+1} $ can not mapped to $ p_{i+1} $, because $ p_{i-1} $ and $ p_{i+1} $ is not visible from each other. Let $ c_{j+1} $ is mapped to $ p_{t} $ and Let $ c_{j-1} $ is mapped to $ p_{s} $. Now, replace edges $ \left({p}_{s},{p}_{i-1}\right) $ and $ \left({p}_{i-1},{p}_{t}\right) $ in $ C_{m} $ with $ \left({p}_{s},{p}_{i}\right) $ and $ \left({p}_{i},{p}_{t}\right) $. The cycle graph $ C_{m} $ will contain reflex vertex $ p_{i} $ with the size as same as being before replacement. See figure~\ref{fig:figure09}(a). For every reflex vertices $ p_{v} $ in this situation do same operation till it remain no reflex vertex like that. In the second case, let there is a node, as denoted $ c_{j} $, that is mapped to $ p_{i+1} $. Let $ c_{j+1} $ and $ c_{j-1} $ are mapped to $ p_{s} $ and $ p_{t} $. Now, replace edges $ \left({p}_{s},{p}_{i+1}\right) $ and $ \left({p}_{i+1},{p}_{t}\right) $ in $ C_{m} $ with $ \left({p}_{s},{p}_{i}\right) $ and $ \left({p}_{i},{p}_{t}\right) $. The cycle graph $ C_{m} $ will contain reflex vertex $ p_{i} $ with the size as same as being before replacement. See figure~\ref{fig:figure09}(b). For every reflex vertices $ p_{w} $ in second explained situation do same operation till it remain no reflex vertex like that. Finally, $ C_{m} $ will be contained all reflex vertices on $ P_{n} $, where $ P_{n} $ is a pseudo-convex polygon.
\end{proof}
 \begin{figure}
  \centering
  \includegraphics[width=0.8\linewidth]{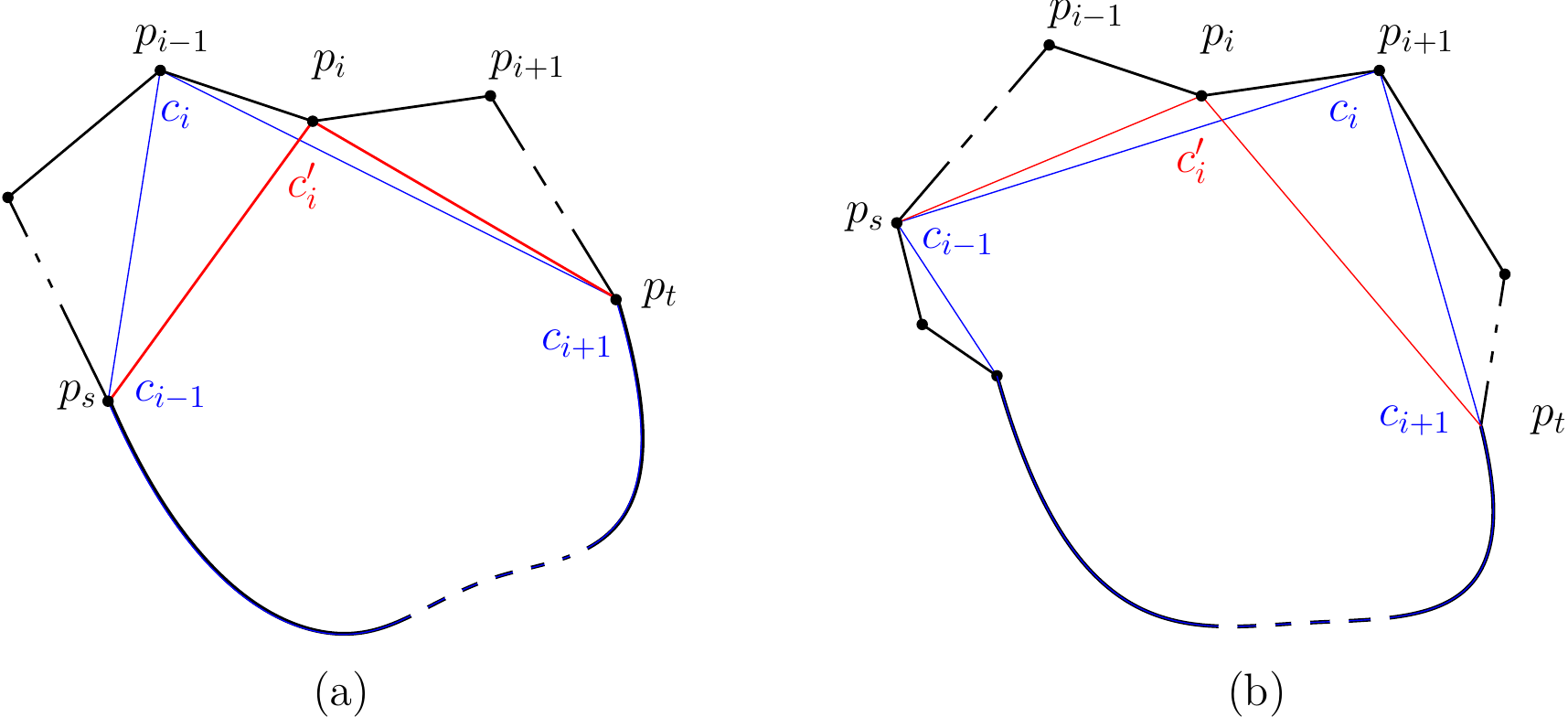}
  \caption{{\small The illumination of the proof}}
  \label{fig:figure09}
  \end{figure}
Let $ C_{m} $ has the property that is explained in lemma~\ref{le:le05}. We provide a linear-time algorithm to find $ C_{m} $. Hence, We have the following theorem.
\begin{theorem}
\label{th:th01}
There is a linear-time algorithm for finding straight-line planar embedding of a cycle graph with maximum size in pseudo-convex polygon $ P_{n} $ with $ n $ vertices such that each node of cycle mapped to a vertex on $ P $.
\end{theorem}
\begin{proof}
Let $ \{p_{0},p_{1},\dots,p_{n-1}\} $ be vertices set of $ P_{n} $ ordered counter clockwise. As lemma~\ref{le:le05}, there is a straight-line planar embedded cycle graph with maximum size such that all reflex vertices of $ P_{n} $ are on it. We provide an algorithm to find such a cycle, as denoted $ C_{m} $, if $ m $ be number of its edges. We will find $ m $ after tracing algorithm. Therefore, all reflex vertices of $ P_{n} $ is selected to being in $ C_{m} $.

Traverse vertices of $ P $ in following order:
\begin{enumerate}
\item Traverse from $ {q}_{m+2} $ to $ {q}_{m} $.
\item Traverse from $ {q}_{m-i} $ to $ {q}_{i} $ {\small (for $ 0\leq i<\frac{m}{2} $)} if current location index is greater than $ \frac{m}{2} $.
\item Traverse from $ {q}_{i} $ to $ {q}_{m-i-1} $ {\small (for $ 0\leq i< \frac{m}{2}-1 $)} if current location index is less than $ \frac{m}{2} $.
  \end{enumerate}
 Because of convexity, the vertices of convex polygon are visible from each others. Therefore, according to the algorithm, path $ P $ has exactly $ m $ edges which are all planar chords of $ Q $. Time complexity is $ \theta \left(m\right) $, see figure~\ref{fig:figure09}.
\end{proof}

\section{Conclusion}
In this paper, we present a linear-time algorithms for embedding the maximum-size path and cycle in a convex polygon with $ n $ vertices. After that, we use a similar algorithm for pseudo-convex polygons which have no u-turn vertex and time complexity is remained linear. But, time complexity of the problem for simple polygon remains open. If it is NP-hard, finding a fix parameter algorithm is interesting for this problem.

\bibliographystyle{plain}
\bibliography{bibfile}

\end{document}